\documentclass[a4paper,12pt]{article}
\RequirePackage[OT1]{fontenc}
\usepackage{amsthm,amsmath,natbib}
\RequirePackage{graphicx}
\RequirePackage{tabularx,caption}
\RequirePackage[colorlinks,citecolor=blue,urlcolor=blue,linkcolor=blue]{hyperref}
\RequirePackage{bm,amsfonts}

\usepackage[onehalfspacing]{setspace}%
\newtheorem{theorem}{\normalfont\scshape Theorem}
\newtheorem{definition}{\normalfont\scshape Definition}
\newtheorem{lemma}{\normalfont\scshape Lemma}
\newtheorem{proposition}{\normalfont\scshape Proposition}

\renewenvironment{proof}{{\normalfont\scshape Proof}}

\begin{document}

\begin{center}
{\Large On estimation in varying coefficient models for sparse and irregularly sampled functional data}
\end{center}

\vspace{.2cm}
\begin{center}

{\large Behdad\ Mostafaiy}\\

Department\ of\ Statistics \\

University\ of\ Mohaghegh\ Ardabili \\

Ardabil,\ Iran \\

behdad.mostafaiy@gmail.com

\end{center}

\begin{abstract}
In this paper, we study a smoothness regularization method for a varying coefficient model based on sparse and irregularly sampled functional data which is contaminated with some measurement errors. We estimate the one-dimensional covariance and cross-covariance functions of the underlying stochastic processes based on a reproducing kernel Hilbert space approach. We then obtain least squares estimates of the coefficient functions. Simulation studies demonstrate that the proposed method has good performance. We illustrate our method by an
analysis of longitudinal primary biliary liver cirrhosis data.
\end{abstract}

\section{Introduction}

Varying coefficient models are introduced by \cite{HT1993}. They are an extension of classical linear regression models where the coefficients are smooth functions. They are used for modeling the dynamic impacts of the underlying covariates on the response. Varying coefficient models have been extensively studied in the literature. Various types of varying coefficient models have been studied and developed for longitudinal data, time series, high dimensional data and functional data. See, for example, \cite{HRWY1998}, \cite{KT1999}, \cite{WC2000}, \cite{CRW2001}, \cite{HWZ2004}, \cite{RS2005}, \cite{SM2010}, \cite{ZLK2012}, \cite{V2014}, \cite{SYZ2014}, \cite{KP2015} and \cite{LM2016} among others.

In this paper, we consider the following multiple varying coefficient model
\begin{equation}\label{vary}
Y(t) = \beta_0 (t) + \sum \limits_{p=1}^{d_1}\beta_p (t)X_p(t)+ \sum \limits_{q=1}^{d_2}\alpha_q(t)Z_q + \eta(t),\qquad t\in \mathcal{T}
\end{equation}
where $Y(t)$ is the response process, $X_1(t),\dots ,X_{d_1}(t)$ are the predictor processes, $Z_1,\dots ,Z_{d_2}$ are time-independent predictors, $\eta(t)$ is a noise process with zero mean and independent of the predictors, and $\beta_0 (t)$, $\beta_1(t),\dots ,\beta_{d_1}(t)$ and $\alpha_1(t),\dots ,\alpha_{d_2}(t)$ are smoothed parameter functions. It is assumed that $Y(t)$ and $X_1(t),\dots , X_{d_1}(t)$ are square integrable and $Z_1,\dots ,Z_{d_2}$ have finite second moments.

The aim of this article is estimating the parameter functions in the situation that the observations are sparse and irregular longitudinal data and combined with some measurement errors. Following \cite{YMW2005}, we model this situation as follows. Let $U_{ij}$ and $V_{ij}$ denote the observations of the random functions $X_i$ and $Y_i$ respectively at the random times $T_{ij}$, contaminated with measurement errors $\varepsilon _{pij}$ and $\epsilon_{i\,j}$ respectively, which are assumed to be independent and identically distributed with means zero and variances $\sigma^2_{_{X_p}}$ and $\sigma^2_{_{Y}}$ respectively, and independent of the random functions. We represent the observed data as
\begin{equation}\label{UV}
\begin{aligned}
U_{pij}&=&X_{pi}(T_{ij})+\varepsilon _{pij}, \qquad j=1,\,\dots,\, M_i; \qquad i=1,\,\dots,\,n\,,\\
V_{ij}&=&Y_i(T_{ij})+\epsilon _{ij},\qquad \quad j=1,\,\dots,\, M_i; \qquad i=1,\,\dots,\,n\,.
\end{aligned}
\end{equation}
Here $M_i$ is a nonnegative integer-valued random variable that denotes the sampling frequency for $i$th trajectory.

For sparse noisy functional data, \cite{SM2010} studied model \eqref{vary} with one functional predictor. They obtained a representation for the coefficient function based on one-dimensional covariance and cross-covariance functions of the predictor and response processes. They used local linear smoother method for their estimation procedures. \cite{SN2011} extended the approach of \cite{SM2010} to multiple predictors including both functional and non-functional predictors. \cite{MRH2016} considered one functional predictor. They proposed a reproducing kernel Hilbert space approach to estimate the coefficient function.

By taking expectation from the both sides of \eqref{vary}, we have
\begin{equation}\label{beta0}
\beta_0 (t) = \mu_Y(t) - \sum \limits_{p=1}^{d_1}\beta_p (t)\mu_{X_p}(t) - \sum \limits_{q=1}^{d_2}\alpha_q(t)\mu_{Z_q},\qquad t\in \mathcal{T},
\end{equation}
where $\mu_Y(t)=E[Y(t)]$, $\mu_{X_p}(t)=E[X_p(t)]$, $p=1,\dots ,d_1$ and $\mu_{Z_q}=E[Z_q]$, $q=1,\dots ,d_2$. Substituting equation \eqref{beta0} in \eqref{vary} yields
\begin{equation}\label{vary2}
Y(t)- \mu_Y(t) = \sum \limits_{p=1}^{d_1}\beta_p (t)(X_p(t)-\mu_{X_p}(t))+ \sum \limits_{q=1}^{d_2}\alpha_q(t)(Z_q-\mu_{Z_q}) + \eta(t),\qquad t\in \mathcal{T}
\end{equation}

By multiplying both sides of \eqref{vary2} by $X_p(t)$, $p=1,\dots ,k_1$ and $Z_q$, $q=1,\dots ,k_2$, and then taking expectations and writing the results in matrix form, we get
\begin{equation}\label{parameters}
[\beta_1(t),\dots ,\beta_{k_1}(t),\alpha_1(t),\dots ,\alpha_{k_2}(t)]^\prime =\bm{ \Gamma } _t^{-1}\bm{\gamma}_t,
\end{equation}
where
\begin{equation*}
\bm{ \Gamma } _t =
\begin{bmatrix}
C_{X_1X_1}(t) & \dots & C_{X_1X_{k_1}}(t) & C_{X_1Z_1}(t) & \dots & C_{X_1Z_{k_2}}(t) \\
\vdots & \ddots & \vdots & \vdots & \ddots & \vdots
\\
C_{X_{k_1}X_1}(t) & \dots & C_{X_{k_1}X_{k_1}}(t) & C_{X_{k_1}Z_1}(t) & \dots & C_{X_{k_1}Z_{k_2}}(t) \\
C_{Z_1X_1}(t) & \dots & C_{Z_1X_{k_1}}(t) & C_{Z_1Z_1} & \dots & C_{Z_1Z_{k_2}} \\
\vdots & \ddots & \vdots & \vdots & \ddots & \vdots
\\
C_{Z_{k_2}X_1}(t) & \dots & C_{Z_{k_2}X_{k_1}}(t) & C_{Z_{k_2}Z_1} & \dots & C_{Z_{k_2}Z_{k_2}} \\
\end{bmatrix},
\end{equation*}
and
\begin{equation*}
\bm{\gamma}_t =
\begin{bmatrix}
C_{YX_1}(t) & \dots & C_{YX_{k_1}}(t) & C_{YZ_1}(t) & \dots & C_{YZ_{k_2}}(t)
\end{bmatrix} ^ \prime .
\end{equation*}
Here $C_{X_{p_1}X_{p_2}}(t)=\operatorname{cov}(X_{p_1}(t),X_{p_2}(t))$, $C_{X_{p}Z_{q}}(t)=C_{Z_{q}X_{p}}(t)=\operatorname{cov}(X_p(t),Z_q)$, $C_{Z_{q_1}Z_{q_2}}=\operatorname{cov}(Z_{q_1},Z_{q_2})$,
$C_{YX_p}(t)=\operatorname{cov}(Y(t),X_p(t))$ and $C_{YZ_q}(t)=\operatorname{cov}(Y(t),Z_q)$. Based on the representation \eqref{parameters}, we introduce an estimate of the parameter functions. To do this, we estimate every elements of $\bm{\Gamma}_t$ and $\bm{\gamma}_t$. The scalar parameters of $\bm{\Gamma}_t$ can be easily estimated. To estimate the parameter functions of $\bm{\Gamma}_t$ and $\bm{\gamma}_t$, we use a reproducing kernel Hilbert space (RKHS) framework. By assuming the sample paths of $X_p$s, $p=1,\dots ,k_1$, and $Y$ to be smooth such that they belong to some RKHSs, we show that the one-dimensional covariance and cross-covariance functions come from some RKHSs. Based on these results, we introduce some smoothness regularization methods to estimate these parameter functions. By simulation, we investigate the merits of the proposed method especially by comparing it to some other existing methods.

The paper is organized as follows. In Section \ref{section2}, we review some basic properties of RKHS. In Section \ref{section3}, we utilize a regularization method to estimate the one-dimensional covariance and cross-covariance functions and then provide estimates of the coefficient functions. Simulation studies in two cases (one predictor and multiple predictors) are provided in Section \ref{section4}. In Section \ref{section5}, we apply the method to longitudinal primary biliary liver cirrhosis data.


\section{Reproducing kernel Hilbert spaces}\label{section2}
The theory of RKHS plays a pivotal role in this paper. In this section, we present some fundamental concepts and basic facts of RKHS. The readers are referred to \cite{A1950}, \cite{BTA2004} and \cite{HE2015} for more details.
\begin{definition}\label{D1}
A symmetric, real-valued bivariate function $K$ on $\mathcal{T}\times \mathcal{T}$ is nonnegative definite, denoted by $K \geq 0$, provided that
\begin{equation*}
\sum \limits_{i=1}^N \sum \limits_{j=1}^N \alpha_i \alpha_j K(t_i,t_j) \geq 0,
\end{equation*}
for all $N \in \mathbb{N}$, $\alpha _1, \ldots , \alpha _N \in \mathbb{R}$, and $t_1, \ldots ,t_N \in \mathcal{T}$. In other words, $K \geq 0$ provided that for every $N \in \mathbb{N}$ and distinct points, $\{t_1,\ldots ,t_N \} \subseteq \mathcal{T}$, the matrix $\mathbf{K}:=[K(t_i,t_j)]$ be a nonnegative definite matrix, that is $\mathbf{K} \geq 0$.
\end{definition}
\begin{lemma}\label{L1}
Let $\mathcal{H}$ is a Hilbert space with inner product $\langle \cdot , \cdot \rangle _{\mathcal{H}}$ and $\phi : \mathcal{T}\longrightarrow \mathcal{H}$ is a function on $\mathcal{T}$. Then the function $K(s,t):=\langle \phi (s),\phi (t) \rangle _{\mathcal{H}}$ on $\mathcal{T} \times \mathcal{T}$ is nonnegative definite.
\end{lemma}
\begin{definition}\label{D2}
For a Hilbert space $\mathcal{H}$ with inner product $\langle \cdot ,\cdot \rangle _{\mathcal{H}}$, a bivariate function $K(s,t)$ for $s,t \in \mathcal{T}$ is called a reproducing kernel of $\mathcal{H}$ if the following are satisfied:
\begin{itemize}
\item[(i)] For every $t \in \mathcal{T}$, $K(\cdot,t) \in \mathcal{H}$.
\item[(ii)] For every $t \in \mathcal{T}$ and every $f \in \mathcal{H}$, 
\begin{equation}\label{RP}
f(t)=\langle f , K(\cdot , t) \rangle _{\mathcal{H}}.
\end{equation}
\end{itemize}
\end{definition}
Relation \eqref{RP} is called the reproducing property of $K$.
\begin{definition}\label{D3}
A Hilbert space $\mathcal{H}$ of functions on $\mathcal{T}$ is called an RKHS if there exist a reproducing kernel $K$ of $\mathcal{H}$.
\end{definition}
From now on, we denote a reproducing kernel Hilbert space $\mathcal{H}$ with the reproducing kernel $K$ by $\mathcal{H}(K)$ and the corresponding inner product and norm by $\langle \cdot , \cdot \rangle _{\mathcal{H}(K)}$ and $\| \cdot \|_{\mathcal{H}(K)}$, respectively.

By using properties (i) and (ii) in Definition \ref{D2}, for any $N,N^{\prime} \in \mathbb{N}$, $\alpha _1, \ldots , \alpha _N,\alpha _1^{\prime}, \ldots , \alpha _{N^{\prime}}^{\prime} \in \mathbb{R}$ and $t_1, \ldots ,t_N,t_1^{\prime}, \ldots ,t_{N^{\prime}}^{\prime} \in \mathcal{T}$, we have
\begin{equation}\label{NormSum}
\langle \sum\limits_{i=1}^N \alpha_i K(\cdot,t_i),\sum\limits_{j=1}^{N^{\prime }} \alpha_j^{\prime} K(\cdot,{t^{\prime }}_j) \rangle _{\mathcal{H}(K)} =
\sum\limits_{i=1}^N \sum\limits_{j=1}^{N^{\prime }} \alpha_i \alpha_j^{\prime} K(t_i,{t^{\prime }}_j)
\end{equation}
The following proposition states the uniqueness of reproducing kernel $K$ and RKHS $\mathcal{H}(K)$.
\begin{proposition}\label{P1}
If $K$ is a reproducing kernel of $\mathcal{H}(K)$ then $K$ is nonnegative definite and unique. Conversely, if $K$ is a nonnegative definite bivariate function on $\mathcal{T} \times \mathcal{T}$, there exists a uniquely determined Hilbert space $\mathcal{H}(K)$ of functions on $\mathcal{T}$, admitting the reproducing kernel $K$.
\end{proposition}
In the next proposition, we give a condition which characterizes the function that belong to an RKHS.
\begin{proposition}\label{P2}
A real-valued function $f$ defined on $\mathcal{T}$ belongs to the reproducing kernel Hilbert space $\mathcal{H}(K)$ if and only if there exists a constant $C$ such that, $C^2K(s,t)-f(s)f(t)$ is a nonnegative definite function on $\mathcal{T} \times \mathcal{T}$, i.e. $C^2K(s,t)-f(s)f(t) \geq 0$.
\end{proposition}
Let $\mathcal{H}(K_1 \otimes K_2):=\mathcal{H}(K_1)\otimes \mathcal{H}(K_2)$ is the tensor product Hilbert space of $\mathcal{H}(K_1)$ and $\mathcal{H}(K_2)$, where $\mathcal{H}(K_1)$ and $\mathcal{H}(K_2)$ are two RKHSs of functions defined on $\mathcal{T}$ with reproducing kernels $K_1$ and $K_2$ respectively. 
Consider the map $\phi :T\longrightarrow \mathcal{H}(K_1 \otimes K_2)$ defined by $\phi (t)(\cdot ,*)=(K_1 \otimes K_2)((\cdot ,*),(t,t))$. Then, for $s,t \in \mathcal{T}$,
\begin{align*}
K_1(s,t)K_2(s,t) &= \langle K_1(\cdot , s),K_1(\cdot ,t) \rangle _{\mathcal{H}(K_1)} \langle K_2(* , s),K_2(* ,t) \rangle _{\mathcal{H}(K_2)} \\
&= \langle (K_1 \otimes K_2)((\cdot ,*),(s,s)),(K_1 \otimes K_2)((\cdot ,*),(t,t)) \rangle _{\mathcal{H}(K_1 \otimes K_2)} \\
&= \langle \phi (s),\phi (t) \rangle _{\mathcal{H}(K_1 \otimes K_2)}.
\end{align*}
Therefore by Lemma \ref{L1}, the pointwise product of two reproducing kernel $K_1$ and $K_2$ is nonnegative definite and so it is a reproducing kernel by Proposition \ref{P1}. So we can construct the RKHS $\mathcal{H}(K_1K_2)$ uniquely. In particular, if $K$ is reproducing kernel of $\mathcal{H}(K)$ then $K^2$ is reproducing kernel of $\mathcal{H}(K^2)$.

The following Theorem is fundamental for estimation procedures in the next section.
\begin{theorem}\label{T1}
Suppose that $X$ and $Y$ are two stochastic processes such that the sample paths of $X$ and $Y$, respectively, belong to $\mathcal{H}(K_1)$ and  $\mathcal{H}(K_2)$ almost surely and $E\|X\|_{\mathcal{H}(K_1)}^2 < \infty$ and $E\|Y\|_{\mathcal{H}(K_2)}^2 < \infty$. Then
\begin{itemize}
\item[(i)] $\mu _X $ and $\mu _Y$ belong to $\mathcal{H}(K_1)$ and  $\mathcal{H}(K_2)$ respectively.
\item[(ii)] $C_{XX}$ and $C_{XY}$ belong to $\mathcal{H}(K_1^2)$ and $\mathcal{H}(K_1K_2)$ respectively. 
\end{itemize}
\end{theorem}
The proof is based on the following Lemma.
\begin{lemma}\label{L2}
Let $\mathbf{A}$ and $\mathbf{B}$ are two $N$ dimensional matrices and $\mathbf{A} \circ \mathbf{B}$ denotes the Hadamard product of $\mathbf{A}$ and $\mathbf{B}$.
\begin{itemize}
\item[(i)] If $\mathbf{A} \geq 0$ and $\mathbf{B} \geq 0$ then $\mathbf{A} \circ \mathbf{B} \geq 0$.
\item[(ii)] If $\mathbf{A} \geq \mathbf{B} \geq 0$ then $\mathbf{A} \circ \mathbf{A} \geq \mathbf{B} \circ \mathbf{B}$.
\end{itemize}
\end{lemma}
\begin{proof}
(i) Let $\mathbf{A}=[a_{ij}]$, $\mathbf{B}=[b_{ij}]$ and $T=\{1,\ldots ,N\}$. Suppose that $f_1$ and $f_2$ are two functions on $T \times T$ such that $f_1(i,j)=a_{ij}$ and $f_2(i,j)=b_{ij}$, $(i,j)\in T\times T$. Then $f_1$ and $f_2$ are nonnegative definite functions. Because the pointwise product of two nonnegative definite functions is again nonnegative definite, we have $\mathbf{A} \circ \mathbf{B} \geq 0$.
\\
(ii) We have $\mathbf{A}+\mathbf{B} \geq 0$ and $\mathbf{A}-\mathbf{B} \geq 0$. By part (i) of this Lemma, $(\mathbf{A}+\mathbf{B})\circ (\mathbf{A}-\mathbf{B}) \geq 0$ and so $\mathbf{A}\circ \mathbf{A} -\mathbf{B}\circ \mathbf{B}=(\mathbf{A}+\mathbf{B})\circ (\mathbf{A}-\mathbf{B}) \geq 0$ or $\mathbf{A}\circ \mathbf{A}\geq \mathbf{B}\circ \mathbf{B}$. \qed
\end{proof}

\noindent
\textsc{Proof} of \textsc{Theorem \ref{T1}}. By Jensen's inequality, we have
\begin{equation*}
\|\mu _X\|^2_{\mathcal{H}(K_1)}\leq E\|X\|^2_{\mathcal{H}(K_1)}<\infty \qquad \text{and} \qquad \|\mu _Y\|^2_{\mathcal{H}(K_2)}\leq E\|Y\|^2_{\mathcal{H}(K_2)}<\infty,
\end{equation*}
which complete proof of (i). To prove (ii), we only show that $C_{XY} \in \mathcal{H}(K_1K_2)$, as $C_{XX} \in \mathcal{H}(K_1^2)$ is an immediate consequence of $C_{XY} \in \mathcal{H}(K_1K_2)$. Let $s,t \in \mathcal{T}$. First notice that
\begin{align*}
C_{XY}(t) &= E[X(t)Y(t)]-\mu _X(t)\mu _Y(t) \\
&=: \mu _{XY}(t)-(\mu _X \mu _Y)(t).
\end{align*}
Because $X \in \mathcal{H}(K_1)$ almost surely, by Proposition \ref{P2}, there exists a constant $C_1$ such that
\begin{equation}\label{C1X}
C_1^2K_1(s,t)-X(s)X(t) \geq 0, \qquad a.s.
\end{equation}
Similarly, there exists a constant $C_2$ such that
\begin{equation}\label{C2Y}
C_2^2K_2(s,t)-Y(s)Y(t) \geq 0, \qquad a.s.
\end{equation}
Therefore Lemma \ref{L2} together with the equations \eqref{C1X} and \eqref{C2Y} imply that
\begin{equation*}
(C_1C_2)^2(K_1K_2)(s,t)-[X(s)Y(s)][X(t)Y(t)] \geq 0, \qquad a.s.
\end{equation*}
Now, Proposition \ref{P2} implies that $XY$ belongs to $\mathcal{H}(K_1K_2)$ almost surely and therefore by part (i) of this Theorem, $\mu _{XY}\in \mathcal{H}(K_1K_2)$. It remains to show that $\mu _X \mu _Y \in \mathcal{H}(K_1K_2)$. Part (i) of this Theorem and Proposition \ref{P2} implies that there exists constants $C_3$ and $C_4$ such that
\begin{equation*}
C_3^2K_1(s,t)-\mu _X(s)\mu _X(t) \geq 0
\end{equation*}
and
\begin{equation*}
C_4^2K_2(s,t)-\mu _Y(s)\mu _Y(t) \geq 0.
\end{equation*}
So, by Lemma \ref{L2},
\begin{equation*}
(C_3C_4)^2(K_1K_2)(s,t)-[\mu _X(s)\mu _Y(s)][\mu _X(t)\mu _Y(t)] \geq 0.
\end{equation*}
Now Proposition \ref{P2} implies that $\mu _X \mu _Y \in \mathcal{H}(K_1K_2)$.
\qed


\section{Estimation Methods}\label{section3}
In this section, we introduce estimates of the parameters involved in \eqref{beta0} and \eqref{parameters}. Assume that the sample paths of $Y$ and $X_{p}$ for $p=1,\dots ,k_1$ respectively belong to $\mathcal{H}(K)$ and $\mathcal{H}(K_p)$ almost surely, where $\mathcal{H}(K)$ and $\mathcal{H}(K_p)$ are some RKHSs. Since $Z_q$s are time-independent, a natural estimate for $\mu_{Z_q}$ is $\hat{\mu}_{Z_q}=\bar{Z}_q=\dfrac{1}{n}\sum \limits_{i=1}^nZ_{qi}$. Also the mean functions $\mu_Y(t)$ and $\mu_{X_p}(t)$ can be estimated by either of the methods given in \cite{YMW2005}, \cite{LH2010}, \cite{CY2011} and \cite{ZW2016}. Denote the estimated mean functions of $Y$ and $X_p$ by $\hat{\mu}_Y(t)$ and $\hat{\mu}_{X_p}(t)$ respectively. The covariance $C_{Z_{q_1}Z_{q_2}}$ can be simply estimated by 
$\hat{C}_{Z_{q_1}Z_{q_2}}=\dfrac{1}{n}\sum \limits_{i=1}^n(Z_{q_1i}-\bar{Z}_{q_1})(Z_{q_2i}-\bar{Z}_{q_2})$. 
To estimate the one-dimensional covariance and cross-covariance functions, define the raw covariance terms
\begin{align*}
C_{X_{p_1}X_{p_2},ij}(T_{ij})&=[ U_{p_1ij}  - \hat{\mu}_{X_{p_1}} (T_{ij})][ U_{p_2ij}  - \hat{\mu}_{X_{p_2}} (T_{ij})], \\
C_{YX_p,ij}(T_{ij})&=[ V_{ij}  - \hat{\mu}_Y (T_{ij})][ U_{pij}  - \hat{\mu }_{X_p} (T_{ij})], \\
C_{X_pZ_q,ij}(T_{ij})&=[ U_{pij}  - \hat{\mu}_{X_p} (T_{ij})][Z_{qi}-\bar{Z}_q], \\
C_{YZ_q,ij}(T_{ij})&=[ V_{ij}  - \hat{\mu}_Y (T_{ij})][Z_{qi}-\bar{Z}_q].
\end{align*}
By Theorem \ref{T1}, $C_{X_{p1}X_{p2}}\in \mathcal{H}(K_{p_1}K_{p_2})$, $C_{X_pZ_q}\in \mathcal{H}(K_p)$, $C_{YX_p}\in \mathcal{H}(KK_p)$ and $C_{YZ_q}\in \mathcal{H}(K)$. Based on these results, we estimate the one-dimensional covariance and cross-covariance functions as follows:
\begin{itemize}
\item Estimate of $C_{X_{p1}X_{p2}}$. Define
\begin{equation}\label{CXX}
\hat{C}_{X_{p_1}X_{p_2}}=\mathop {\arg \min }\limits_{C  \in \mathcal{H}(K_{p_1}K_{p_2})} \left\{ \ell _{_{X_{p_1}X_{p_2}}} (C) + \lambda _{_{X_{p_1}X_{p_2}}} \, \|C \|^2_{\mathcal{H}(K_{p_1}K_{p_2})} \right\},
\end{equation}
where 
\begin{equation*}
\ell _{_{X_{p_1}X_{p_2}}} (C) = \frac{1}
{n}\sum\limits_{i = 1}^n \frac{1}{M_i} \sum\limits_{j=1}^{M_i} \left\lbrace  C_{X_{p_1}X_{p_2},ij}(T_{ij}) - C(T_{ij})\right\rbrace^2,
\end{equation*}
and $\lambda _{_{X_{p_1}X_{p_2}}}$ is a smoothing parameter.
\item Estimate of $C_{X_pZ_q}$. Define
\begin{equation}\label{CXZ}
\hat{C}_{X_pZ_q}=\mathop {\arg \min }\limits_{C  \in \mathcal{H}(K_p)} \left\{ \ell _{_{X_pZ_q}} (C) + \lambda _{_{X_pZ_q}} \, \|C \|^2_{\mathcal{H}(K_p)} \right\},
\end{equation}
where 
\begin{equation*}
\ell _{_{X_pZ_q}} (C) = \frac{1}
{n}\sum\limits_{i = 1}^n \frac{1}{M_i} \sum\limits_{j=1}^{M_i} \left\lbrace  C_{X_pZ_q,ij}(T_{ij}) - C(T_{ij})\right\rbrace^2,
\end{equation*}
and $\lambda _{_{X_pZ_q}}$ is a smoothing parameter.
\item Estimate of $C_{YX_p}$. Define
\begin{equation}\label{CYX}
\hat{C}_{YX_p}=\mathop {\arg \min }\limits_{C  \in \mathcal{H}(KK_p)} \left\{ \ell _{_{YX_p}} (C) + \lambda _{_{YX_p}} \, \|C \|^2_{\mathcal{H}(KK_p)} \right\},
\end{equation}
where 
\begin{equation*}
\ell _{_{YX_p}} (C) = \frac{1}
{n}\sum\limits_{i = 1}^n \frac{1}{M_i} \sum\limits_{j=1}^{M_i} \left\lbrace  C_{YX_p,ij}(T_{ij}) - C(T_{ij})\right\rbrace^2,
\end{equation*}
and $\lambda _{_{YZ_q}}$ is a smoothing parameter.
\item Estimate of $C_{YZ_q}$. Define
\begin{equation}\label{CYZ}
\hat{C}_{YZ_q}=\mathop {\arg \min }\limits_{C  \in \mathcal{H}(K)} \left\{ \ell _{_{YZ_q}} (C) + \lambda _{_{YZ_q}} \, \|C \|^2_{\mathcal{H}(K)} \right\},
\end{equation}
where 
\begin{equation*}
\ell _{_{YZ_q}} (C) = \frac{1}
{n}\sum\limits_{i = 1}^n \frac{1}{M_i} \sum\limits_{j=1}^{M_i} \left\lbrace  C_{YZ_q,ij}(T_{ij}) - C(T_{ij})\right\rbrace^2,
\end{equation*}
and $\lambda _{_{YZ_q}}$ is a smoothing parameter.
\end{itemize}
Now, we explain how the minimization problem \eqref{CXX} can be solved. The solutions of \eqref{CXZ}, \eqref{CYX} and \eqref{CYZ} are obtained similarly. Following the representer theorem (see \cite{W1990}),  we consider $C_{X_{p_1}X_{p_2}}$ as the form
\begin{equation}
C_{X_{p_1}X_{p_2}}(t)=\sum \limits_{i=1}^n\sum \limits_{j=1}^{M_i} a_{ij}K_{p_1}(t,T_{ij})K_{p_2}(t,T_{ij})
\end{equation}
for some vector $\mathbf{a}=[a_{11},\dots ,a_{1M_{1}},\dots ,a_{n1},\dots ,a_{nM_{n}}]^\prime$. Now by equation \eqref{NormSum} we have
\begin{align*}
\|C_{X_{p_1}X_{p_2}}\|^2_{\mathcal{H}(K_{p_1}K_{p_2})} 
&= \sum \limits_{i=1}^n\sum \limits_{j=1}^{M_i} \sum \limits_{i^\prime =1}^n\sum \limits_{j^\prime =1}^{M_{i^\prime}} a_{ij}a_{i^\prime j^\prime} K_{p_1}(T_{i^\prime j^\prime},T_{ij})K_{p_2}(T_{i^\prime j^\prime},T_{ij}) \\
&= \mathbf{a}^\prime \mathbf{Q}\mathbf{a},
\end{align*}
where
\begin{equation*}
\mathbf{Q} = \begin{pmatrix}
 \mathbf{Q}_{11} &  \mathbf{Q}_{12}&  \mathbf{Q}_{13} &\cdots &  \mathbf{Q}_{1n}  \\
 \mathbf{Q}_{21} &  \mathbf{Q}_{22} &  \mathbf{Q}_{23} &\cdots&  \mathbf{Q}_{2n} \\
 \vdots &\vdots &\ddots &\vdots &\vdots \\
 \mathbf{Q}_{n1} & \mathbf{Q}_{n2} & \mathbf{Q}_{n3}&\ldots & \mathbf{Q}_{nn}  \\ 
\end{pmatrix}
\end{equation*}
and for $i_1,i_2=1,\dots ,n$, the $(i_1,i_2)$ partition of $\mathbf{Q}$, that is $\mathbf{Q}_{i_1 i_2}$, is an $M_{i_1}\times M_{i_2}$ dimensional matrix with entries $K_{p_1}(T_{i_1 j_1},T_{i_2 j_2})K_{p_2} (T_{i_1 j_1},T_{i_2 j_2})$.
Define
\begin{equation*} 
\mathbf{g} = [g_{11},\dots ,g_{1M_1},\dots ,g_{n1},\dots ,g_{nM_n}]^\prime ,
\end{equation*}
where
\begin{equation*}
g_{ij}=C_{X_{p_1}X_{p_2},ij}(T_{ij}),\qquad i=1,\dots ,n, \quad j=1,\dots ,M_i.
\end{equation*}
Suppose $\|\cdot \|^2_{_F}$ represents the Frobenius norm. Then 
\begin{equation}\label{equivalent1}
\ell _{2} (C_{X_{p_1}X_{p_2}}) + \lambda _2 \, \|C_{X_{p_1}X_{p_2}} \|^2_{\mathcal{H}(K_{p_1}K_{p_2})}=\dfrac{1}{n} \|\mathbf{m}\circ \mathbf{g} -  \mathbf{m}\circ (\mathbf{Q}\mathbf{a}) \|^2_{_F} + \lambda_{X_{p_1}X_{p_2}}  \mathbf{a}^{\prime} \mathbf{Q} \mathbf{a},
\end{equation}
where
$\mathbf{m}=[\frac{1}{\sqrt{M_1}}\mathbf{1}^\prime_{_{M_1}},\dots ,\frac{1}{\sqrt{M_n}}\mathbf{1}^\prime_{_{M_n}}]^\prime$ and $\mathbf{1}_M$ is an $M$ dimensional vector with all one entry. So to solve the minimization problem \eqref{CXX}, it suffices to find a vector $\mathbf{a}$ that minimizes the right hand side of \eqref{equivalent1}. It is not hard to show that the minimizer of right hand side of \eqref{equivalent1} is
\begin{equation*}
\mathbf{a}= \left( \mathbf{P}+( \sum \limits_{i=1}^nM_i) \lambda_{X_{p_1}X_{p_2}} \mathbf{I}\right)^{-1}(\mathbf{m}\circ \mathbf{m}\circ \mathbf{g}),
\end{equation*}
where $\mathbf{P}=\mathbf{Q}\circ (\mathbf{1}^\prime _{_{\sum \limits_{i=1}^nM_i}} \otimes \mathbf{m}) \circ (\mathbf{1}^\prime _{_{\sum \limits_{i=1}^nM_i}} \otimes \mathbf{m})$.

The plug-in estimators of the intercept and coefficient functions are given by
\begin{equation*}
[\hat{\beta}_1(t),\dots ,\hat{\beta}_{k_1}(t),\hat{\alpha}_1(t),\dots ,\hat{\alpha}_{k_2}(t)]^\prime =\hat{\bm{ \Gamma }} _t^{-1}\hat{\bm{\gamma}}_t
\end{equation*}
and
\begin{equation*}
\hat{\beta}_0 (t) = \hat{\mu}_Y(t) - \sum \limits_{p=1}^{d_1}\hat{\beta_p} (t)\hat{\mu}_{X_p}(t) - \sum \limits_{q=1}^{d_2}\hat{\alpha}_q(t)\hat{\mu}_{Z_q}.
\end{equation*}


\section{Simulation studies} \label{section4}
In this section, we evaluate the performance of the proposed method. We provide two simulation examples. In the first simulation, we consider one functional predictor and compare our method, denoted by LSRK, with the methods given in \cite{SM2010} and \cite{MRH2016}. In the second simulation, we consider two functional and one time-independent predictors and compare our method with the method of \cite{SN2011}. The methods of \cite{SM2010} and \cite{SN2011} are implemented in the \texttt{MATLAB} package \texttt{PACE} which can be downloaded from the website \url{http://www.stat.ucdavis.edu/PACE/}. In the all simulation studies, we consider $\mathcal{T}=[0,1]$. To face with sparse and irregular situation, we generated uniformly the number of measurements for each trajectory from $\{4,5,6,7,8\}$ and the random locations $T_{ij}$s from $\mathcal{T}$.

As in \cite{SN2011}, we measure the estimation accurracy by mean absolute deviation error ($\operatorname{MADE}$) and weighted average squared error ($\operatorname{WASE}$) defined by
$$
\operatorname{MADE}=\frac{1}{d_1d_2}\left[ \sum \limits_{p=1}^{d_1}\dfrac{\int_0^1|\beta_p(t)-\hat{\beta}_p(t)|dt}{\operatorname{range}(\beta_p)} + \sum \limits_{q=1}^{d_2}\dfrac{\int_0^1|\alpha_q(t)-\hat{\alpha}_q(t)|dt}{\operatorname{range}(\alpha_q)} \right] 
$$
and
$$
\operatorname{WASE}=\frac{1}{d_1d_2}\left[ \sum \limits_{p=1}^{d_1}\dfrac{\int_0^1(\beta_p(t)-\hat{\beta}_p(t))^2dt}{\operatorname{range}^2(\beta_p)} + \sum \limits_{q=1}^{d_2}\dfrac{\int_0^1(\alpha_q(t)-\hat{\alpha}_q(t))^2dt}{\operatorname{range}^2(\alpha_q)} \right].
$$
All integrals numericaly computed by Gaussian quadrature method.  

We consider various combinations of the sample size $n\in \{100,\,150,\,200\}$ and the signal-to-noise ratio $\operatorname{StN}\in \{4,\,8,\,\infty \}$. For each configuration, we repeat the experiment $500$ times.
\subsection{Simulation study 1}
The random function $X_1$ was generated as
$$
X_1(t)=\mu_{X_1}(t)+\sum \limits_{k=1}^{50}a_k\xi_k \phi_k(t),
$$
where 
$$
\mu_{X_1}(t)=\sum \limits_{k=1}^{50}(-1)^kk^{-3/2} \phi_k(t),
$$
$$
a_k=4(-1)^k/k^2,
$$
and
$$
\phi_k(t)=\sqrt{2}\cos (2k\pi t)\,.
$$
The marginal distributions of $\xi_1,\dots ,\xi_{50}$ are $N(0,1)$. Observations from process $X(t)$ were obtained by adding measurement errors $U_{1ij}=X_{1i}(T_{ij})+\varepsilon_{ij}$, where $\varepsilon_{ij}$s were independently generated from $N(0,\sigma^2_{_{X_1}})$ with $\sigma^2_{_{X_1}}=(4.2954/\operatorname{StN})^2$.

In the model \eqref{vary} with only one predictor $X_1$, we consider $\beta_0(t)=2\sin (2\pi t)$ and $\beta_1(t)=2e^t$. The sparse and noisy response observations were obtained by $V_{ij}=\beta_0(T_{ij})+\beta_1(T_{ij})U_{iij}+\epsilon_{ij}$, where the noise terms $\epsilon_{ij}$s randomly drawn from $N(0,\sigma^2_{_{Y}})$ with $\sigma^2_{_{Y}}=(15.6815/\operatorname{StN})^2$. 

Table \ref{tab1} presents the Monte Carlo values of MADE and WASE for the three competitive methods LSRK (proposed), \cite{SM2010} and \cite{MRH2016}. Although the method of \cite{MRH2016} outperforms other two methods but it is slightly better than LSRK. From this Table, we observe that LSRK has significantly better performance than the method of \cite{SM2010}. The performance of LSRK is improved by increasing either the sample size or the signal-to-noise ratio. In Figure \ref{MISE1}, we provide the mean integrated squared errors of $\hat{\beta}_0$ and $\hat{\beta}_1$ for the method LSRK. In this Figure, the left panel is for $\hat{\beta}_0$ and the right panel for $\hat{\beta}_1$. We observe that increasing both the sample size $n$ and the signal-to-noise ratio $\operatorname{StN}$ lead to accurate estimates. This improvement is more significant when $\operatorname{StN}$ is large.
\begin{table}[h!]
\centering 
{\scriptsize \begin{tabular}{|cc|cc|cc|cc|}
\hline
 & & \multicolumn{2}{c}{LSRK} & 
\multicolumn{2}{c}{\cite{SM2010}} & \multicolumn{2}{c|}{\cite{MRH2016}}\\
\cline{3-8}
$n$ & $\operatorname{StN}$ & MADE & WASE & MADE & WASE & MADE & WASE  \\
\hline
&&&&&&&\\
 & $4$ & $0.4366$ & $0.4083$ & $0.7273$ & $ 1.8194$ & $0.3109$ & $0.2576$  \\
 &&&&&&&\\
\cline{3-8}
&&&&&&&\\
  $100$ & $8$ & $0.2477$ & $0.1736$ & $0.6053$ & $1.3622$ & $0.1197$ & $0.0446$\\
  &&&&&&&\\
\cline{3-8}
&&&&&&&\\
   & $\infty$ & $0.2238$ & $0.1434$ & $0.6213$ & $ 1.3932$ & $0.0895$ & $0.0302$ \\
   &&&&&&&\\
\hline
&&&&&&&\\
 & $4$ & $0.3806$ & $0.3270$ & $ 0.7143$ & $4.3996$& $0.2666$ & $0.1867$ \\
 &&&&&&&\\
\cline{3-8}
&&&&&&&\\
 $150$  & $8$ & $0.2036$ & $0.1109$ & $0.6470$ & $2.0066$& $0.1083$ & $0.0371$ \\
 &&&&&&&\\
\cline{3-8}
&&&&&&&\\
   & $\infty$ & $0.1905$ & $0.1015$ & $0.6235$ & $1.1665$ & $0.0886$ & $0.0293$ \\
   &&&&&&&\\
\hline
&&&&&&&\\
 & $4$ & $0.3527$ & $0.2739$ & $0.7028$ & $1.5030$ & $0.2388$ & $0.1487$ \\
 &&&&&&&\\
\cline{3-8}
&&&&&&&\\
 $200$  & $8$ & $0.1817$ & $0.0857$ & $0.6670$ & $1.5393$& $0.1055$ & $0.0354$ \\
 &&&&&&&\\
\cline{3-8}
&&&&&&&\\
   & $\infty$ & $0.1669$ & $0.0766$ & $0.6400$ & $1.3321$& $0.0855$& $0.0259$\\
   &&&&&&&\\
\hline
\end{tabular}}
\caption{Mean absolute deviation error (MADE) and weighted average squared error (WASE) for various combinations of sample size ($n$) and signal-to-noise ratio ($\operatorname{StN}$). The compared three methods are: LSRK (proposed), \cite{SM2010}, and \cite{MRH2016}.}
\label{tab1}
\end{table}
\begin{figure}[h!]
\centering
\includegraphics[scale=.75]{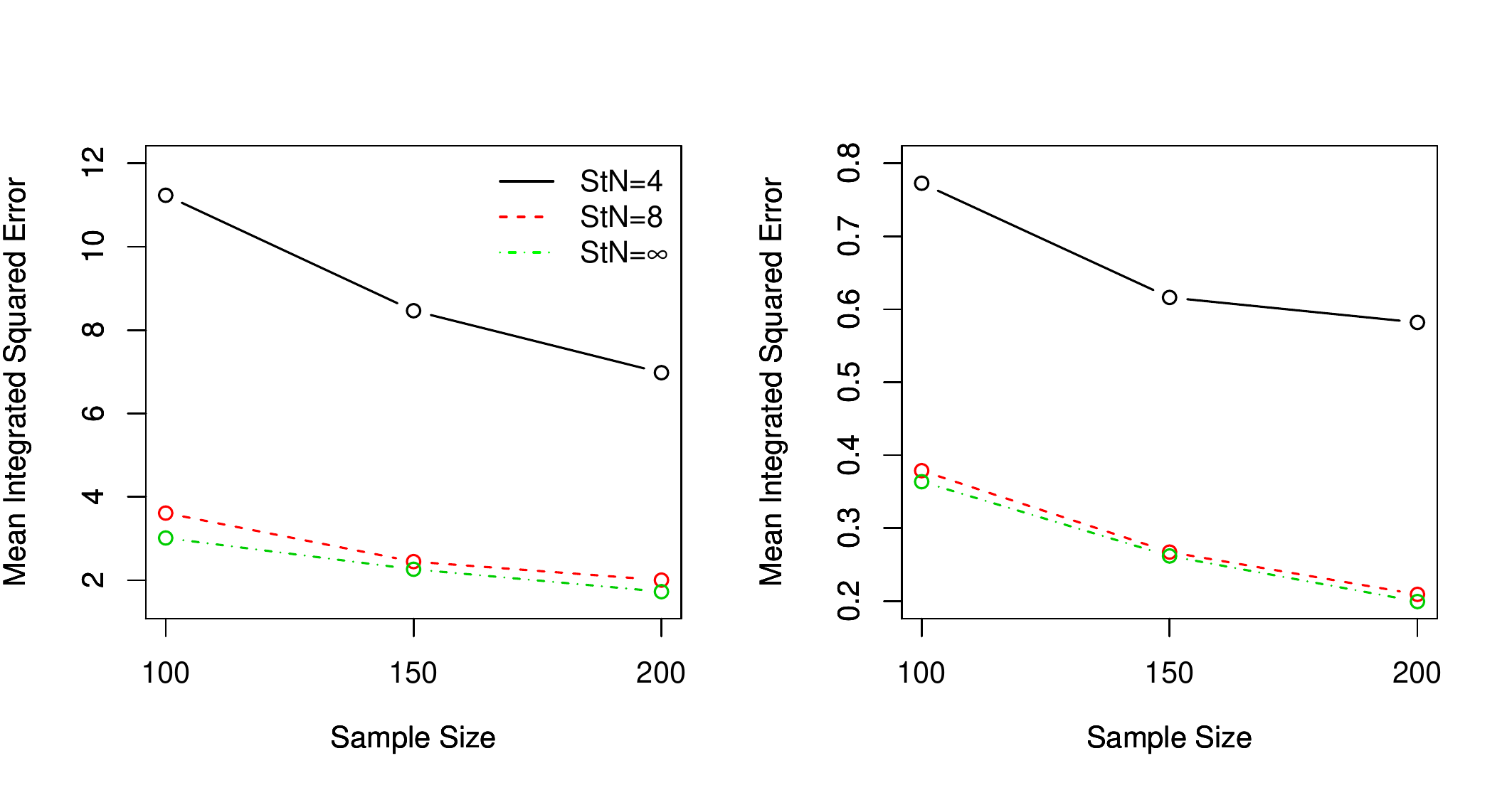}
\caption[]{Effect of signal-to-noise ratio and sample size on integrated squared errors of $\hat{\beta}_0$ (left panel) and $\hat{\beta}_1$ (right panel) for the method LSRK.}
\label{MISE1}
\end{figure}

\subsection{Simulation study 2}
The first functional predictor is same as previous subsection. For the second functional predictor, we took
$$
X_2(t)=\mu_{X_2}(t)+\sum \limits_{k=1}^{50}b_k\zeta_k \psi_k(t),
$$
where
$$
\mu_{X_2}(t)=\sin(2\pi t)-te^{-t}\, ,
$$
$$\psi_k(t)=\begin{cases}
\sqrt{2}\sin (2k\pi t) \qquad& \text{ for } k\leq 49\\
1 \qquad& \text{ for }k=50
\end{cases}$$
and
$$b_k=\begin{cases}
\sqrt{3}/2^k \qquad& \text{ for } k\leq 49\\
\sqrt{3} \qquad& \text{ for }k=50 \,.
\end{cases}$$
Also $\zeta_1,\dots ,\zeta_{50}$ are marginally distributed as $N(0,1)$. Sparse and noisy observations $U_{2ij}$s from random function $X_2$ were obtained based on model \eqref{UV}, where $\varepsilon_{2ij}$s were independent distributed as $N(0,\sigma^2_{_{X_2}})$ with $\sigma^2_{_{X_2}}=(1.2733/\operatorname{StN})^2$. The marginal distribution of the time-independent covariate $Z$ is $N(1,1)$. To have correlation between the predictors, let $\bm{\Sigma}=[\Sigma_{kl}]$ be the covariance matrix of the random vector $[Z,\,\xi_1,\,\dots ,\,\xi_{50},\,\zeta_1,\,\dots ,\,\zeta_{50}]^{\prime}$, where
$$
\Sigma_{kl}=\begin{cases}
1 \qquad & \text{ for } k=l \\
0.4^{l-1} \qquad & \text{ for } k=1,\,2\leq l\leq 51 \\
(-0.3)^{l-51} \qquad & \text{ for } k=1,\,l\geq 52 \\
0.8^{l-50} \qquad & \text{ for } k=l-50,\, l\geq 52 \\
0 \qquad & \text{ otherwise}\,.
\end{cases}
$$
The response observations $V_{ij}$s were obtained from
$$
V_{ij}=\beta_0(T_{ij})+\beta_1(T_{ij})U_{1ij}+\beta_2(T_{ij})U_{2ij}+\alpha_1(T_{ij})Z_i+\epsilon_{ij},
$$
where random errors $\epsilon_{ij}$s were independently generated from $N(0,\sigma^2_{_{Y}})$ with $\sigma^2_{_{Y}}=(15.8525/\operatorname{StN})^2$. Also $\beta_0(t)$ and $\beta_1(t)$ are same as simulation study 1, and $\beta_2(t)=5te^{-t}$ and $\alpha_1(t)=2t$.

We compare LSRK with the method of \cite{SN2011}. Table \ref{tab2} summarizes the Monte Carlo values of $\operatorname{MADE}$ and $\operatorname{WASE}$ for two methods. In all combinations of $n$ and $\operatorname{StN}$, LSRK has the smallest values of MADE and WASE. Moreover, LSRK appears to be more stable. As expected, increasing either sample size $n$ and signal-to-noise ration $\operatorname{StN}$ decreases estimation errors. Figure \ref{MISE2} displays mean integrated squared errors of the estimated coefficient functions, the top left panel for $\hat{\beta}_0$, the top right panel for $\hat{\alpha}_1$, the bottom left panel for $\hat{\beta}_1$ and the bottom right panel for $\hat{\beta}_2$. This Figure reveals that there is a general tendency for the mean integrated squared errors to decrease as either sample size or signal-to-noise ratio increases.
\begin{table}[h!]
\centering 
{\scriptsize \begin{tabular}{|cc|cc|cc|}
\hline
 & & \multicolumn{2}{c}{LSRK} & 
\multicolumn{2}{c|}{\cite{SN2011}} \\
\cline{3-6}
$n$ & $\operatorname{StN}$ & MADE & WASE & MADE & WASE   \\
\hline
&&&&&\\
 & $4$ & $0.5556$ & $0.7703$ & $1.6690$ & $8378.6490$   \\
 &&&&&\\
\cline{3-6}
&&&&&\\
  $100$ & $8$ & $0.3228$ & $0.4987$ & $0.9533$ & $280.0505$ \\
  &&&&&\\
\cline{3-6}
&&&&&\\
   & $\infty$ & $0.2918$ & $0.3709$ & $0.8757$ & $412.0377$  \\
   &&&&&\\
\hline
&&&&&\\
 & $4$ & $0.4808$ & $0.4899$ & $1.3142$ & $1177.7210$ \\
 &&&&&\\
\cline{3-6}
&&&&&\\
 $150$  & $8$ & $0.2587$ & $0.3119$ & $0.8899$ & $98.4621$ \\
 &&&&&\\
\cline{3-6}
&&&&&\\
   & $\infty$ & $0.2220$ & $0.2771$ & $0.9590$ & $279.5293$ \\
   &&&&&\\
\hline
&&&&&\\
 & $4$ & $0.4255$ & $0.3679$ & $1.1581$ & $312.3150$ \\
 &&&&&\\
\cline{3-6}
&&&&&\\
 $200$  & $8$ & $0.2197$ & $0.2809$ & $0.9329$ & $636.4345$ \\
 &&&&&\\
\cline{3-6}
&&&&&\\
   & $\infty$ & $0.1922$ & $0.2254$ & $0.9388$ & $675.6874$ \\
   &&&&&\\
\hline
\end{tabular}}
\caption{Mean absolute deviation error (MADE) and weighted average squared error (WASE) for various combinations of sample size ($n$) and signal-to-noise ratio ($\operatorname{StN}$). The compared two methods are: LSRK (proposed), and \cite{SN2011}.}
\label{tab2}
\end{table}
\begin{figure}[h!]
\centering
\includegraphics[scale=.75]{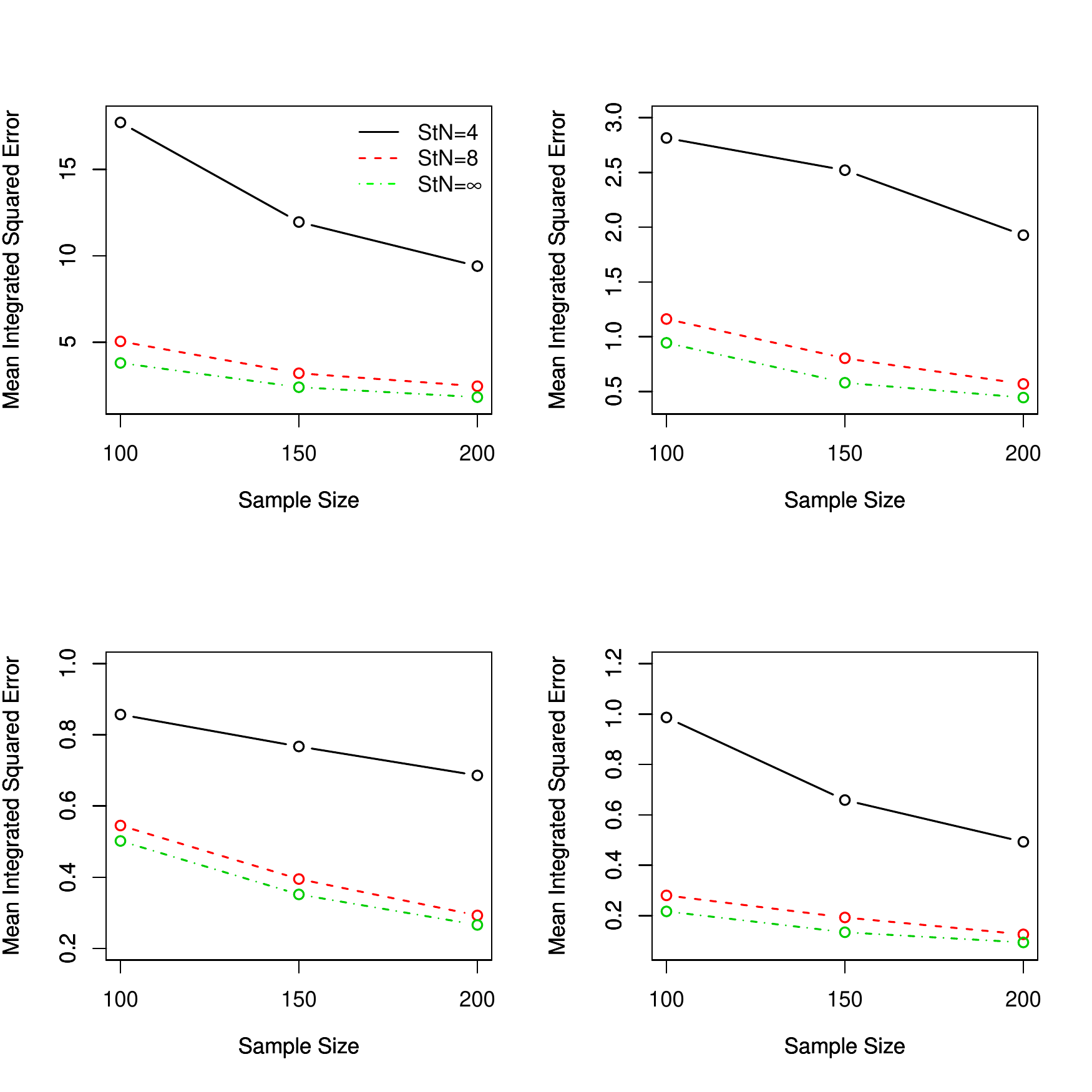}
\caption[]{Effect of signal-to-noise ratio and sample size on integrated squared errors of $\hat{\beta}_0$ (top left panel), $\hat{\alpha}_1$ (top right panel), $\hat{\beta}_1$ (bottom left panel) and $\hat{\beta}_2$ (bottom right panel) for the method LSRK.}
\label{MISE2}
\end{figure}


\section{Application}\label{section5}
Primary biliary cirrhosis (PBC) is an autoimmune liver disease. It caused by damage to the bile (a fluid produced in the liver to aid in the digestion of fat) ducts in the liver. When the bile ducts are damaged,  bile builds up and causes liver scarring, cirrhosis, and eventually liver failure. The dataset that we use in this paper was collected by the Mayo Clinic between $1974$ and $1984$.  The dataset is given in Appendix D of \cite{FH1991} and also included in the \texttt{R} package \texttt{survival} which is available at \url{https://cran.r-project.org/package=survival}. The patients were scheduled to have their blood characteristics measured at six months, one year and annually after diagnosis. Because of missing appointments, death or liver transplantation during the study and other factors, the actual times of the measurements are random, irregular and sparse.

This dataset contains some general information for example age in days and sex, and some multiple laboratory results for example serum bilirubin in mg$/$dl, albumin  in gm$/$dl and prothrombin time in seconds. 
Bilirubin is a yellow substance that is formed during the normal breakdown of red blood cells. After circulating in the blood, the liver excretes bilirubin into bile ducts. The normal adult serum bilirubin level is less than $1$ mg$/$dl. The accumulation of bilirubin leads to jaundice. Albumin is a protein made by the liver. It is the main protein in the blood that causes fluid to remain within the bloodstream. A diseased liver produces insufficient albumin. The normal albumin range is $3.5$ to $5.5$ g$/$dl. 
 Prothrombin time is the time it takes for blood to clot. Liver disease can cause slow blood clotting. The average time range for prothrombin time is about $10$ to $14$ seconds.

The objective of this analysis is to explore the association between prothrombin time ($Y$) as a response and age ($Z_1$), serum bilirubin ($X_1$) and albumin ($X_2$) as predictors. Among $276$ female patients, we include $137$ patients having D-penicillamine and their measurements before $2500$ days. The median number of observations per patients is $5$. Individual trajectories and data along with the smoothed estimated mean functions of prothrombin time, bilirubin and albumin are given in Figure \ref{data}. The mean prothrombin time slightly increases by passing time but it is normal. The mean amount of bilirubin is above the normal level and it has an increasing trend. By passing the time, the mean amount of albumin made by the liver decreases.
\begin{figure}[h!]
\centering
\includegraphics[scale=.75]{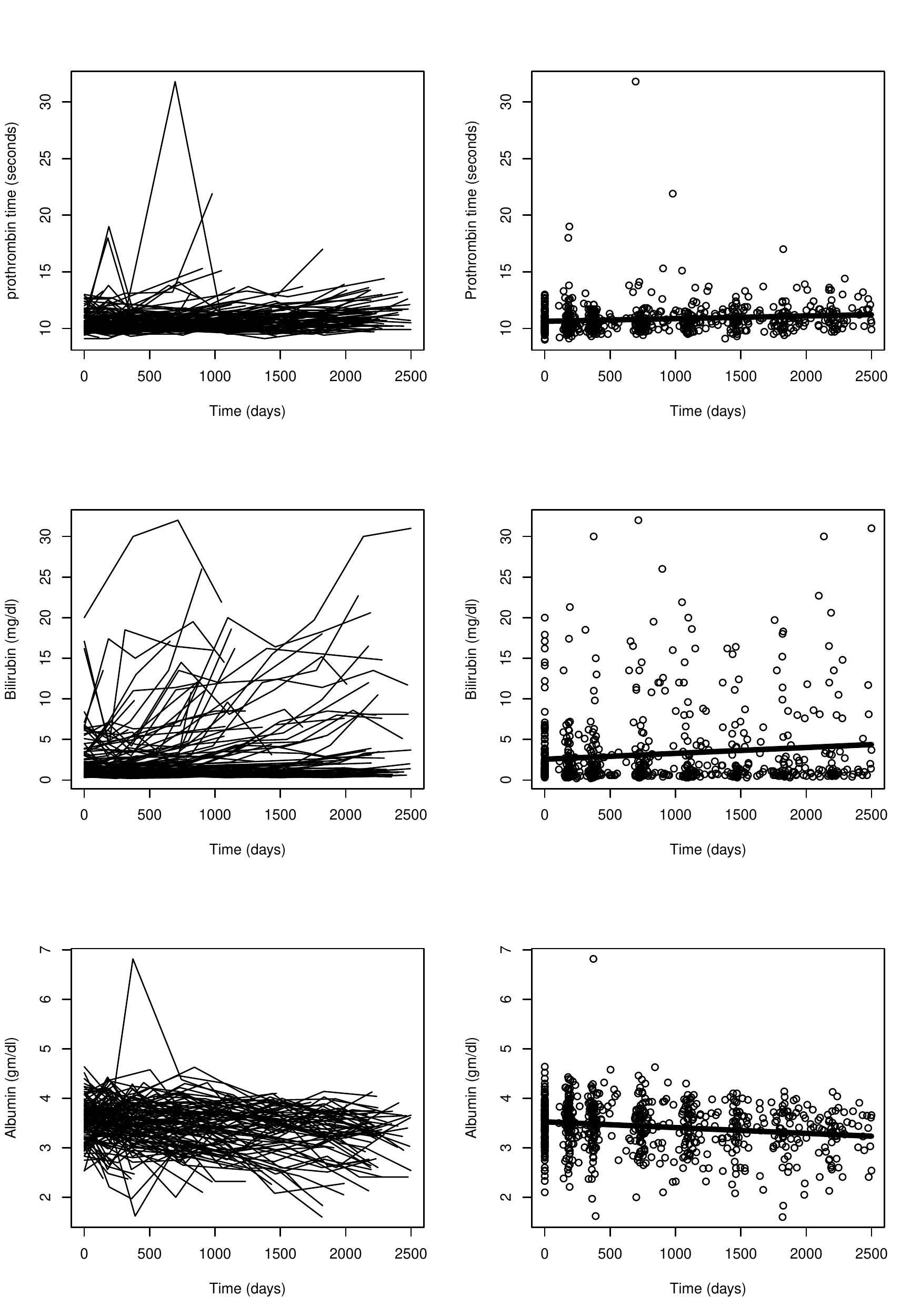}
\caption[]{The left panels give the observed individual trajectories, the top panel for prothrombin time, the middle panel for bilirubin, and the lower panel for albumin.  The observed data along with the estimated mean functions (solid line) are shown in the right panels, the top panel for prothrombin time, the middle panel for bilirubin, and the lower panel for albumin.}
\label{data}
\end{figure}

Figure \ref{coef} plots the estimated varying coefficient functions $\beta_0$, $\alpha_1$, $\beta_1$ and $\beta_2$ using LSRK. From the Figure we observe that before $2000$ days the association between age and prothrombin time is negligible but after $2000$ days age has a negative effect on prothrombin time. There exists a negative association between albumin and prothrombin time, especially after $2000$ days. The effect of  bilirubin on prothrombin time before $2000$ days is minor and fluctuates between positive and negative while after $2000$ days the association tends to be negative.
\begin{figure}[h!]
\centering
\includegraphics[scale=.75]{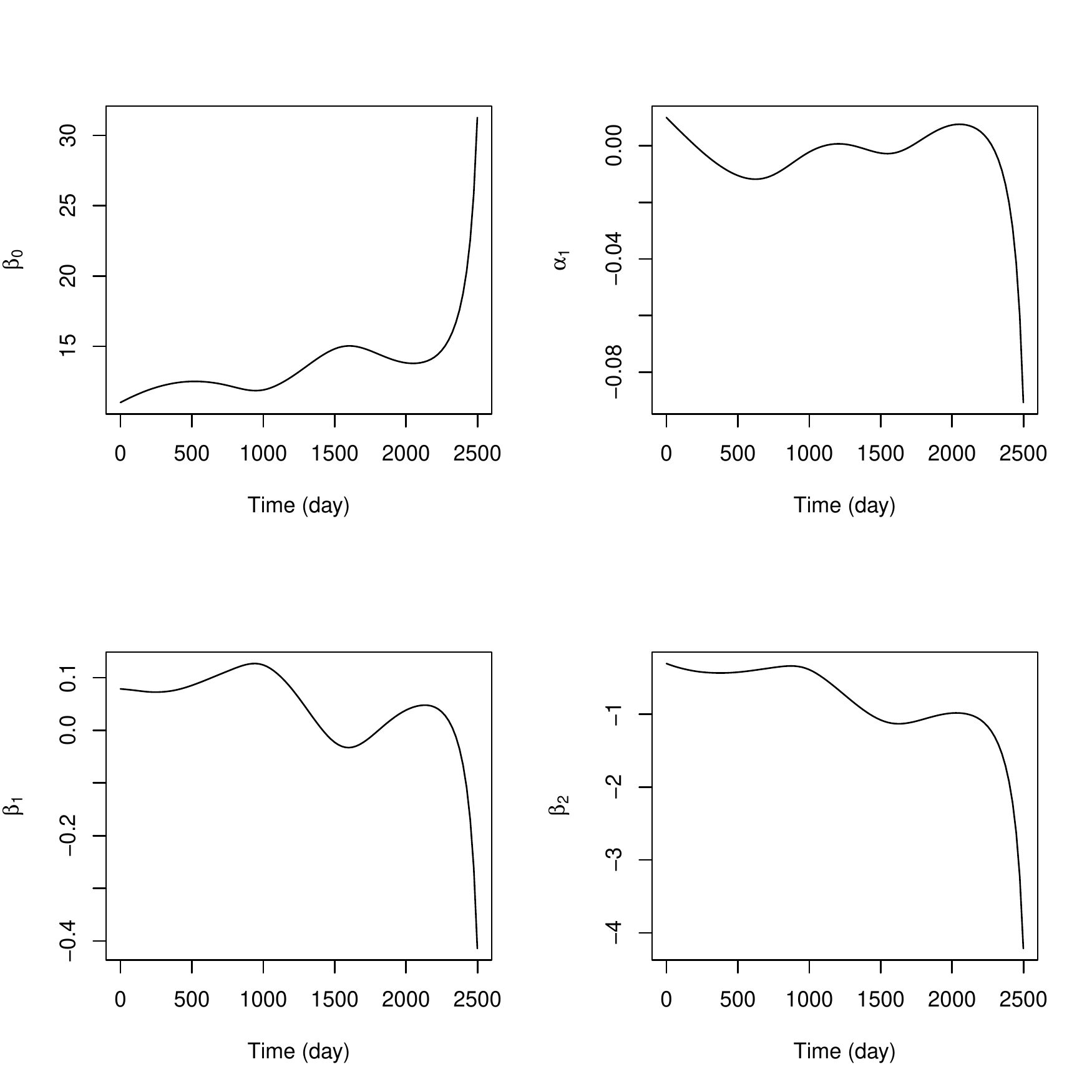}
\caption[]{The estimated varying coefficient functions $\beta_0$, $\alpha_1$, $\beta_1$ and $\beta_2$.}
\label{coef}
\end{figure}


\end{document}